\documentclass[11pt,a4paper,USenglish]{article}
\usepackage{hyperref}
\hypersetup{
    colorlinks=true,
    linkcolor=blue
    }

\usepackage{microtype}

\usepackage{array}
\usepackage{cite} 
\usepackage{mathtools}                
\usepackage[mathlines]{lineno}
\usepackage{amsfonts}
\usepackage{amsthm}
\usepackage{xcolor}
\usepackage[capitalize]{cleveref}
\usepackage{graphicx}

\def\reals{{\mathbb R}}
\def\eps{{\varepsilon}}

\def\Vor{{\rm Vor}}

\def\dist{{\rm dist}}

\def\NN{{\rm N}}
\def\D{{\cal D}}

\def\NN{{\rm N}}
\def\dd{{\sf d}}
\def\dist{{\sf dist}}
\def\prev{{\sf prev}}

\def\optdist{{\sf optdist}}
\def\Vor{{\rm Vor}}
\def\UPDATE{{\sf UPDATE}}

\def\R{{\cal R}}
\def\K{{\cal K}}
\def\U{{\cal U}}
\def\D{{\cal D}}

\newcommand{\old}[1]{{}}

\title{The Unweighted and Weighted Reverse Shortest Path Problem for Disk Graphs\thanks{
Work by H. Kaplan was partially supported by Grant 1595/19 from the Israel Science Foundation and by the Blavatnik Research Foundation. Work by M. Katz was partially supported by Grant 2019715/CCF-20-08551 from the US-Israel Binational Science Foundation/US National Science Foundation. Work by M. Sharir was partially supported by Grant 260/18 from the Israel Science Foundation.
A preliminary version of this paper appears in {\it Proc. European Sympos. Algorithms (ESA)}, 2023.} 
}

\author{Haim Kaplan\thanks{%
    School of Computer Science, Tel Aviv University, Tel Aviv, Israel;
    {\sf haimk@tauex.tau.ac.il}}
\and
Matthew J. Katz\thanks{%
    Department of Computer Science, Ben Gurion University, Beer Sheva, Israel; 
    {\sf matya@cs.bgu.ac.il}}
\and
Rachel Saban\thanks{%
    Department of Computer Science, Ben Gurion University, Beer Sheva, Israel;
    {\sf rachelfr@post.bgu.ac.il}}
\and
Micha Sharir\thanks{%
    School of Computer Science, Tel Aviv University, Tel Aviv Israel;
    {\sf michas@tauex.tau.ac.il}}
}

\newtheorem{theorem}{Theorem}

\newtheorem{lemma}{Lemma}
\newtheorem{claim}{Claim}

\begin{document}

\date{}

\maketitle


\begin{abstract}
We study the reverse shortest path problem on disk graphs in the plane. In this problem we consider the 
proximity graph of a set of $n$ disks in the plane of arbitrary radii: In this graph two disks are 
connected if the distance between them is at most some threshold parameter $r$. The case of intersection graphs is a special case with $r=0$.
We give an algorithm that, given a target length $k$, computes the smallest value of $r$ for which 
there is a path of length at most $k$ between some given pair of disks in the proximity graph.
Our algorithm runs in $O^*(n^{5/4})$ randomized expected time, which improves to $O^*(n^{6/5})$ for unit disk graphs, where all the disks have the same radius.\footnote{%
  In this paper the $O^*(\cdot)$ notation hides subpolynomial factors.}
 Our technique is robust and can be applied to many variants of the problem. One significant variant is
 the case of weighted proximity graphs, where edges are assigned real weights equal to the distance 
 between the disks or between their centers, and $k$ is replaced by a target weight $w$; that is, we seek a path whose length is at most $w$. In other variants,
 we want to optimize a parameter different from $r$,
 such as a scale factor of the radii of the disks.
 
 The main technique for the decision version of the problem (determining whether the graph with a given $r$ 
 has the desired property) is based on efficient implementations of BFS (for the unweighted case) and of 
 Dijkstra's algorithm (for the weighted case), using efficient data structures for maintaining 
 the bichromatic closest pair for certain bicliques and several distance functions. The optimization problem 
 is then solved by combining the resulting decision procedure with enhanced variants of the 
 interval shrinking and bifurcation technique of \cite{BFKKS}. 
\end{abstract}

\newpage
\section{Introduction}

In this paper we study  the \emph{reverse shortest path problem} (RSP for short) on graphs defined by disks in the plane. 

In the simplest variant of this problem,  we are given a set $P$ of $n$ points in the plane, 
two designated points $s,t\in P$, a real parameter $r>0$ and an integer $k < n$. Define $G_r$ to be the graph 
$(P,E_r)$, where the edges of $E_r$ are all the pairs $(p,q)$ such that $\|p-q\| \le 2r$.
This is also the intersection graph of the disks of radius $r$ centered at the points of $P$.
In the decision version of the RSP problem we want to determine whether $G_r$ contains a path from $s$ to $t$ with at most $k$ edges. 
In the optimization version, which is the reverse shortest path problem itself, we wish to find the smallest value $r^*$ for which $G_{r^*}$ has this property.
Both versions (decision and optimization) of this problem have received considerable attention during the 
past decade~\cite{CabelloJ15,CS,WZ}. 

\smallskip
\noindent
{\bf Our contributions: }
We give an algorithm for this problem that runs in $O^*(n^{6/5})$ randomized expected time  (where the $O^*(\cdot)$ notation hides subpolynomial factors). This improves the 
recent $O^*(n^{5/4})$-time solution of Wang and Zhao~\cite{WZ}.
In fact we study the RSP problem in a much more general context that involves a variety of intersection or
proximity graphs on a finite set of arbitrary disks in the plane, for which nontrivial bounds were not known prior to this work.

\looseness = -1
We first consider unweighted disk graphs in Section \ref{sec:unweighted}.
In this setup, we have a set $\D$ of $n$ disks in the plane of arbitrary radii. Each disk $D\in\D$ is specified by its center $c_D$ and radius $\rho_D$. 
For a given parameter $r\ge 0$ we define the \emph{proximity graph} $G_r$ by adding an edge 
between disks $D$ and $D'$ if the distance between them,
$\dist(D,D') = \|c_D - c_{D'}\| - \rho_D - \rho_{D'}$, is at most $r$. The case $r=0$ is of special interest and gives the intersection graph of the disks.\footnote{
One technical difference is that when considering proximity graphs, it is customary, although not obligatory, 
to assume that the disks are pairwise disjoint, which is certainly not an assumption that one would make for intersection graphs.}

For the decision version of this RSP problem we obtain an algorithm that runs in $O(n\log^4n)$ time, 
and for the optimization version an algorithm that runs in $O^*(n^{5/4})$ randomized expected time. Our technique generalizes 
to other versions of the optimization problem. For example, we can consider the intersection graph of the 
disks ($r=0$) and ask for the smallest scaling factor of the radii of all disks, either by a common
additive term or by a common multiplicative factor, that would make the graph contain a path of at most 
$k$ edges between a designated pair of source and target disks $D_s$ and $D_t$.

In Section \ref{sec:weighted} we generalize our results further to weighted versions of the proximity graph $G_r$. We consider two 
natural weight functions for the edges. The first sets the weight of an edge $(D,D')$ to be the distance 
$\|c_D - c_{D'}\|$ between the centers of the disks,
and the second sets the weight to be the distance $\dist(D,D') = \|c_D - c_{D'}\| - \rho_D - \rho_{D'}$
between the disks. We solve the decision problem on such weighted disk graphs, in which we want to determine
whether the shortest path in $G_r$ from $D_s$ to $D_t$ is of length at most $w$, for some specified real 
threshold $w$, by a careful implementation of
Dijkstra's algorithm, following idas introduced by Cabello and Jej\v{c}i\v{c}~\cite{CabelloJ15} (using a dynamic bichromatic closest pair structure, see below) in $O(n\log^4n)$ or 
$O(n\log^6n)$ time, depending on the type of edge weights. The optimization RSP problem is then solved in $O^*(n^{5/4})$ randomized expected time.
For weighted unit disk graphs we still get the better bound of $O^*(n^{6/5})$ time for the optimization problem.

Our decision algorithms rely on rather complex dynamic data structures (see below), which  should be avoided, 
if possible, from a practical point of view. Indeed, for unit disk graphs, there exist simpler and slightly 
more efficient implementations of BFS and Dijkstra's algorithm, in the unweighted and weighted cases, 
respectively~\cite{CabelloJ15,CS,WangX20}. However, as explained in the remarks following Theorems~\ref{thm:gen-unweighted2}~and~\ref{thm:weight2}, we cannot use them in conjunction 
with our optimization technique, for certain technical reasons. We thus present in Section~\ref{sec:udg_bfs_dijk} 
alternative implementations, based on the known grid-based techniques, which satisfy our requirements and are arguably somewhat simpler.        

\smallskip
\noindent
{\bf Our techniques:}
We achieve our results by carefully combining three main  technical ingredients. The first is an 
efficient ``serial'' implementation of parametric search, using what we call \emph{interval shrinking}
and \emph{bifurcation} procedures. This technique was first used by Ben Avraham et al.~\cite{BFKKS} for solving
problems involving the discrete and semi-discrete Fr{\' e}chet distance with shortcuts.
Here we apply a somewhat modified variant of it in the rather different context of
our RSP problem, exposing its potential of being useful for a wide range of other problems as well. 

We remark that using this
technique requires that the decision procedure access the parameter $r$ to be optimized via comparisons only, whose
outcome depends on the relation between the optimal $r^*$ and certain critical values (which in our case 
turn out to be additively weighted inter-point distances between the centers of the disks), 
on which we can apply the interval shrinking procedure in an efficient manner.
We review this technique in Section \ref{sec:preliminaries}.

\looseness=-1
The second ingredient is a dynamic nearest neighbor and a dynamic bichromatic closest pair data structures 
for additively weighted Euclidean distances. Such structures with polylogarithmic time per update 
and access were recently developed by Kaplan et al.~\cite{KaplanMRSS20} and further improved by Liu~\cite{Liu}.

The third ingredient that we use for the weighted versions of the problem is a technique that
combines nearest neighbor data structures for two different distance functions. Specifically, given a 
(dynamic) nearest neighbor data structure for a distance function $d_1$, and a (dynamic) nearest neighbor 
data structure for a distance function $d_2$, we show how we can get a (dynamic) data structure that can answer
constrained nearest neighbor queries of the form: find the closest point to a query $q$ according to 
the distance function $d_1$ among all points whose distance to $q$ according to $d_2$ is at most some threshold $r$ (which is part of the query).

\smallskip
\noindent
{\bf Previous work:}
The decision problem in the unweighted variants can be solved by running a BFS from $s$ (or from $D_s$)
in the underlying graph. Similarly, the decision problem in the weighted variants can be solved by running 
Dijkstra's shortest-path algorithm in the graph. However, the challenge is to do it efficiently, 
since the graph might have up to a quadratic number of edges. For unit-disk graphs, 
Cabello and Jej\v{c}i\v{c}~\cite{CabelloJ15} presented an $O(n\log n)$ implementation of BFS, and subsequently 
Chan and Skrepetos~\cite{CS} presented an alternative $O(n)$ implementation (after pre-sorting the points 
by their $x$- and $y$-coordinates). Moreover, Cabello and Jej\v{c}i\v{c}~\cite{CabelloJ15} also described an
$O(n^{1+\eps})$ implementation of Dijkstra's algorithm for weighted unit-disk graphs, which was followed by a 
more efficient $O(n \log^2 n)$ implementation described by Wang and Xue~\cite{WangX20}; see also~\cite{KaplanMRSS20}. 

The RSP problem in the context of unweighted unit-disk graphs was posed by Cabello and 
Jej\v{c}i\v{c}~\cite{CabelloJ15}, who observed that it can be solved conceptually easily in $O^*(n^{4/3})$ time, 
by running a binary search through the $O(n^2)$ inter-point distances (using an efficient distance selection 
algorithm). Recently, Wang and Zhao~\cite{WZ} managed to improve this bound, obtaining an algorithm that solves 
the problem in $O^*(n^{5/4})$ time. In the context of weighted unit-disk graphs, the situation is similar. 
The RSP problem can be solved easily in $O^*(n^{4/3})$ time, but Wang and Zhao~\cite{WangZ22} were able 
to obtain an improved $O^*(n^{5/4})$-time solution for that version too.\footnote{%
  The $O^*(n^{6/5})$ bound for the RSP problem in both unweighted and weighted unit disk graphs was already claimed in an unpublished manuscript~\cite{KS:opt}, which appeared shortly after the first RSP paper of Wang and Zhao. However, this manuscript overlooks an issue that may arise when using an off-the-shelf decision problem, see Section~\ref{sec:udg_bfs_dijk}.} 

As far as we know, both the decision and optimization problems have not been studied in the context of general disk graphs.

\section{Preliminaries}
\label{sec:preliminaries}

In this section we provide
 necessary background on the serial parametric search technique of Ben Avraham et al.~\cite{BFKKS}.

In its basic form, the technique is applicable when the threshold parameter $r^*$ is the distance between a pair of input points.\footnote{%
  The technique is more broadly applicable, though. It applies in situations where $r^*$ belongs to a set of critical values, each determined by a pair of input objects, so that we have an efficient selection procedure that can find the $k$-th smallest critical value for any input parameter $k$.}
There are $O(n^2)$ such distances, and the most na\"ive algorithm simply finds $r^*$ by running a binary search 
through them, guided by the decision procedure at each comparison. 
A basic improvement is to implement the binary search using an efficient procedure for distance selection, such as the one in \cite{AASS} or \cite{KS}, which runs in $O^*(n^{4/3})$ time.
Up to an additional logarithmic factor, this dominates the cost of the whole procedure 
(assuming that the decision procedure is more efficient; as we show, this is indeed the case in all the RSP problems studied in this paper).

The technique of \cite{BFKKS} is a combination of two subprocedures, referred to as the \emph{interval shrinking} and the \emph{bifurcation} procedures. 
The interval shrinking procedure receives an integer parameter $L\ll \binom{n}{2}$, and computes an interval $I\subset\reals$ that contains 
$r^*$ and at most $L$ \emph{critical values}, namely inter-point distances. As shown in \cite{BFKKS}, this can be done
in $O^*(n^{4/3}/L^{1/3})$ 
expected time.

We then run the bifurcation procedure, which simulates the execution of the decision procedure at the (unknown)
threshold $r^*$, as in the standard parametric search technique~\cite{Megiddo83}.
When the simulation reaches a comparison of $r^*$
with some concrete value $r$, we know the answer to the comparison when $r$ lies outside $I$. However, when $r\in I$
we bifurcate, following both possibilities $r^* > r$ and $r^* < r$ (the case $r^* = r$ will be handled too; see below). This produces a \emph{bifurcation tree} $T$, 
which we expand until we either collect sufficiently many bifurcations, or until we reach a sufficiently large
uniform depth of $T$. In either case we stop this simulation phase, resolve all collected comparisons by a binary search through
them, using the (unsimulated) decision procedure to guide the search, and start a new bifurcation phase from the unique leaf of $T$ 
whose associated (shrunk) interval of critical values contains $r^*$. The binary search will also identify 
$r^*$ when it is one of the critical values through which it searches, and then terminate the entire procedure 
right away. In the worst case this will happen by the time when the entire decision procedure has been simulated.

We comment that this method is viable when the decision procedure is not known to have a parallel version of 
small depth, which is required in the standard parametric search technique. If such a parallel version were 
available, we could apply standard parametric search, and obtain a significantly faster algorithm. The RSP problem
seems to be inherently sequential, as it seeks a path in a graph, and is indeed amenable to the technique of \cite{BFKKS}.

As shown in \cite{BFKKS}, the bifurcation procedure can be implemented to run in $O^*(L^{1/2}D(n))$ time, where $D(n)$
is the cost of the decision procedure. A suitable choice of $L$ yields an overall (randomized expected) running time 
$O^*(n^{6/5})$, for (suitable) decision procedures that run in nearly linear time, as do the decision procedures 
for all the variants of the RSP problem considered in this paper. 

We remark, though that the exponent $6/5$ is obtained in cases where the distance selection procedure, on which 
the interval shrinking is based, runs in $O^*(n^{4/3})$ time. In some of our results we will need distance 
selection in three dimensions, using more general kinds of distance functions. The resulting distance selection
procedure then runs in $O^*(n^{3/2})$ time, and then the overall cost will be $O^*(n^{5/4})$
rather than $O^*(n^{6/5})$; see below for more details.

\section{Reverse shortest paths for unweighted disk graphs}
\label{sec:unweighted}

Here we are given a set $\D$ of $n$ disks in the plane, of arbitrary radii, each  parameterized by its center $c_D$ and radius $\rho_D$. We consider the \emph{intersection graph} $G^\times=(\D,E)$, where the edges of $G^\times$ are the intersecting pairs of disks. 
Formally, $E$ consists of all pairs $(D,D')$ for which
$
\|c_D-c_{D'}\| \le \rho_D + \rho_{D'} .
$

In the decision problem, we are given two designated disks $D_s$ and $D_t$, and an integer parameter $k$, 
and the goal is to determine whether $G^\times$ contains a path of at most $k$ edges from $D_s$ to $D_t$.

In the optimization problem we scale the radii of the disks
(without changing their centers), either by an additive term or by a multiplicative factor, and seek the smallest
scaling parameter that makes $G^\times$ have the desired $s$-$t$ path.

Towards the end of the section, we consider the special and important case of unit disk graphs, for which we obtain a better bound.

\subsection{The decision procedure}

We run BFS on $G^\times$ from $D_s$. Suppose that we have already discovered 
all disks at some level $i$ of the BFS. We expand the BFS to level $i+1$ as follows. We consider each 
disk $D$ at level $i$ in turn, and look for its nearest neighbor among the disks that have not yet been 
discovered. To do so, we maintain a dynamic additively-weighted Voronoi diagram $\Vor(\U)$ of the set $\U$ of 
all the disks that the BFS has not yet reached, where the additive weight of a disk $D$ is $-\rho_D$. 
Initially, we are at level $0$ of the BFS, which includes only $D_s$, and we set $\U = \D\setminus \{D_s\}$.

We search for the nearest neighbor $D'$ of $D$ in $\Vor(\U)$. If the weighted distance between $D$ and 
$D'$ is at most $\rho_D$, that is, if 
$\|c_D-c_{D'}\| - \rho_{D'} \le \rho_D$, we conclude that
$\|c_D-c_{D'}\| \le \rho_D + \rho_{D'}$, so we
add $D'$ to level $i+1$, delete it from $\U$, and query the updated Voronoi diagram again for 
the new nearest neighbor of $D$. We continue querying the updated $\Vor(\U)$ with $D$ until the distance between 
$D$ and its nearest neighbor is larger than $\rho_D$. When this happens we replace $D$ by the next disk at level 
$i$ and repeat this process. When we finish processing in this manner all disks of level $i$, we have 
discovered all disks at level $i+1$, and we move on to level $i+1$.

Consider a sequence of queries to the Voronoi diagram with a disk $D$ at some level of the BFS. 
We can charge each of these queries but the last, to a new disk that we add, following this query,
to the BFS tree. Therefore the running time of this decision procedure is dominated by the cost of
$O(n)$ queries and $n$ deletions from $\Vor(\U)$. The most efficient implementation of such a structure,
with running time $O(n\log^4n)$, is due to Liu~\cite{Liu}; see also the earlier study~\cite{KaplanMRSS20}, with a worse polylogarithmic factor.

In summary, we have shown:
\begin{theorem} \label{thm:dec-gen2}
	Given $\D$, $D_s$, $D_t$, and $k$ as above, we can determine whether there exists a path of at most $k$ edges from 
	$D_s$ to $D_t$ in the intersection graph $G^\times$ associated with $\D$, in $O(n\log^4n)$ time.
\end{theorem}

\subsection{The optimization procedure}

We first consider the case of scaling the radii by an additive term. Thus each disk $D\in\D$ is assigned 
the radius $\rho_D + \alpha$, for some common additive parameter $\alpha$, and we seek the minimum value 
$\alpha^*$ of $\alpha$ for which the intersection graph of the modified disks, now denoted by $G^\times_\alpha$, 
has the desired $s$-$t$ path. In principle $\alpha$ could also be negative, as long as no radius becomes 
negative, but for simplicity we only consider the case $\alpha>0$.

We simulate the decision procedure at the unknown optimal value $\alpha^*$ by using a bifurcation procedure. 
Before starting the simulation, though, we perform an interval-shrinking step, as described in the introduction
and in Section~\ref{sec:preliminaries}. 

\smallskip
\noindent
{\bf Interval shrinking.}
Recall that this step, as introduced in \cite{BFKKS}, receives an integer threshold parameter $L$ and
produces an interval $I_0\subset\reals$ that contains
$\alpha^*$ and at most $L$ other critical values, where a value $\alpha$ is critical if the outcome of
a comparison changes as we go past $\alpha$. In the original formulation, the critical values were
inter-point distances in the plane, and the resulting algorithm ran in $O^*(n^{4/3}/L^{1/3})$ 
randomized expected time. 

Here the setup is different. The comparisons that the decision procedure performs are tests whether expressions of the form
$
\|c_D-c_{D'}\| - \rho_D - \rho_{D'} - 2\alpha
$
are positive or negative. The critical values of the parameter $\alpha$ are thus of the form
$
\frac12 \left( \|c_D-c_{D'}\| - \rho_D - \rho_{D'} \right) .
$
The original mechanism of \cite{BFKKS} is based on distance selection. Here we need a variant in which the basic step is to bound 
the number of these new critical values in a given interval $(\alpha_1,\alpha_2)$. 
We turn this problem into a range searching problem, where the disks of $\D$ serve as both data and query objects. 
Specifically, each disk $D$ is mapped to the point $(c_D,\rho_D)$ in $\reals^3$, and also to the range
$
\sigma_D = \{ D' \mid 
2\alpha_1 \le \|c_D-c_{D'}\| - \rho_D - \rho_{D'} \le 2\alpha_2 \} ,
$
which is a conical shell in 3-space\footnote{%
  It is in fact a truncated conical shell, since we only consider the range $\rho_{D'} \ge 0$.}
(see Figure~\ref{fig:shells}(a)). Note that the ranges have three degrees of freedom, and that the problem 
is symmetric, so that ranges can be represented as points in $\reals^3$ and points as ranges. In other words, we have
a symmetric batched range searching problem in $\reals^3$, involving $n$ points and $n$ semi-algebraic ranges.
Using standard, cutting-based decomposition techniques in $\reals^3$, such as in \cite{Ag:rs,AgarwalKS22}, 
we can implement the range searching step to run in randomized expected time $O^*(n^{3/2})$. Combining this
with parametric search, as in the standard distance selection procedure~\cite{AASS},
we can implement the distance selection procedure to also run in randomized expected time $O^*(n^{3/2})$.

\begin{figure}[hbt]
	\centering
	\includegraphics[scale=0.7]{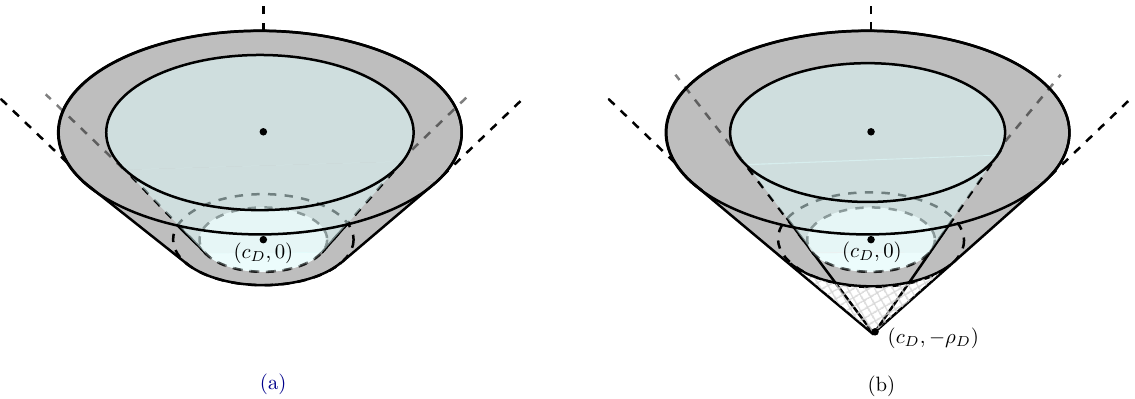}
	\caption{The range $\sigma_D$ (in grey) when scaling by an additive term (a) and by a multiplicative factor (b). The inner and outer radii of the annulus centered at $(c_D,0)$ are $\rho_D+2\alpha_1$ and $\rho_D+2\alpha_2$, respectively, in (a) and $\rho_D\lambda_1$ and $\rho_D\lambda_2$, respectively, in (b).}
	\label{fig:shells}
\end{figure}

Extending the machinery in \cite{BFKKS}, we can convert this distance selection technique to an interval-shrinking
procedure that receives a parameter $L \ll \binom{n}{2}$ and yields an interval $I = (\alpha_1,\alpha_2)$ that
contains the optimum value $\alpha^*$ and  at most $L$ critical values. A suitable modification of the
analysis in \cite{BFKKS} shows that the algorithm runs in 
$O^*(n^{3/2}/L^{1/2})$ randomized expected time. 

\smallskip
\noindent
{\bf Bifurcation.}
We now present the basic bifurcation procedure. This procedure, sometimes with a 
few enhancements and modifications, is used in all our algorithms for the various variants of the RSP problem.
We refer the reader to \cite{BFKKS} where a similar procedure has been used.

Our simulation proceeds in \emph{phases}, where in each phase we construct a bifurcation tree $T$ 
that represents some portion of the execution of the BFS, simultaneously for all values of $\alpha$ in some interval $I$. 
Initially $I = I_0$, but it will keep on shrinking as the simulation proceeds. 
Each node $b$ of $T$ is associated with an interval $I^b = (\alpha_1,\alpha_2) \subseteq I$, such that, up to the state of simulation
represented by $b$, the BFS proceeds in an identical manner for all values $\alpha\in I^b$. We continue to simulate the BFS
at $b$. At each comparison of (the unknown) $\alpha^*$ with some concrete value $\alpha$, we either resolve the comparison in a
unique manner when $\alpha\notin I^b$, or else bifurcate, meaning that we create two children $u$ and $v$ of $b$, 
assign $I^u := (\alpha_1,\alpha)$ and $I^v := (\alpha,\alpha_2)$, and continue to expand $T$ 
at $u$ and at $v$, at each of which we know the outcome of the above procedure 
(the possibility that $\alpha^* = \alpha$ will be tested later).

For each node $b$, let $y_b$ denote the amount of work performed at $b$ so far by the simulation (including comparisons that were fully 
resolved),
and let $y_b^-$ denote the sum of the quantities $y_a$ over all proper ancestors $a$ of $b$, 
\emph{excluding the root}. We refer to the cost at the root as the \emph{initialization cost} of the tree, and denote it as $C_0(T)$.

We stop the expansion of $T$ at a node $b$ when
$y_b^- + y_b = Y$, where $Y$ is a threshold parameter that will be determined later. That is,
we stop the simulation at node $b$ as soon as $y_b^-$ plus the (potentially incomplete) work done so far at $b$ 
becomes $Y$. We refer to such nodes $b$ as \emph{incomplete leaves} of $T$. 
We then continue the expansion of $T$ at other nodes. 

A subtle issue, that we will address later in more detail, is that when we back up from an incomplete node $b$
to explore other branches of $T$, we need to restore the state of execution at the suitable ancestors of $b$. See below for details.

We stop the entire construction of $T$ as soon as
one of the following two conditions occurs:
\begin{description}
\item{(i)}
We collect $X$ bifurcations, for another threshold parameter $X$ that will be determined later.
\item{(ii)}
All the leaves of $T$ are incomplete.
\end{description}
When either of these two conditions happens, we take all the (at most $X$) critical values of the 
bifurcations at the inner nodes of $T$, and run a binary search through them, using
the unsimulated decision procedure to guide the search. This takes $O(D(n)\log n)$ time, and yields the leaf $w$ whose range $I^w$ contains $r^*$.
It is also possible that the binary search will detect that one of the critical values it searches through is $r^*$ 
itself. In this case the entire procedure is terminated, and $r^*$ is output.
Otherwise, this ends a phase of the simulation. We start a new phase (if the simulation has not already ended) with $w$
as the root of the tree, and with $I^w$ as the critical interval. Note that this will cause the work already done
at $w$ to be repeated, and it is possible that the cost of this work is much larger than $Y$. However, by charging the
incomplete work at $w$ to the initialization cost $C_0(T')$ of the next tree $T'$, we at most double this cost, so
this will not affect the overall asymptotic bound 
on the performance of the procedure; see below for details.

\smallskip
\noindent
{\bf Restoring the execution state.}
The issue of restoring the state at nodes we back up to, as mentioned earlier, is more acute in our specific 
application, because part of this state includes the Voronoi diagram that our procedure maintains dynamically.
Although there are persistence-based techniques that can efficiently maintain all versions of the diagram,
they are fairly involved, and we opt not to use them.
Instead, we use the following simple approach, which also takes care of all aspects of restoring the state.

Specifically, we expand the bifurcation
tree in a depth-first manner, so that at each node we first recursively construct the subtree of its left child 
and only then the subtree of its right child. We thus first expand the leftmost path of $T$, and slowly proceed
to the right, backing from a node to its parent, and then proceed to the right child, or further up to the 
grandparent, and so on. When we back up from a node $w$ to its parent $v$, we simply undo the operations performed at $w$,
including the updates that were done to the diagram, 
in reverse order, replacing each deletion of a disk by the corresponding reinsertion. This requires us to maintain
a log of the updates performed at each node, as well as a log of the other operations, but is otherwise
a reasonably simple procedure, which does not affect the asymptotic running time and storage bounds.

\subsection{Analysis}

By construction, a single phase produces a tree $T$ that has
at most $X$ binary nodes, and the cost of producing each path of $T$ is at most $Y$, ignoring the initialization 
cost $C_0(T)$. This is easily seen to imply that the cost of generating $T$ is $O(XY + C_0(T))$. Indeed, the cost at
each node of $T$, other than the root, is certainly at most $Y$, and there are $O(X)$ such nodes 
The cost of the subsequent binary search through the critical values is $O(D(n)\log n)$. Hence the overall cost of a single phase is
$
O(XY + C_0(T) + D(n)\log n) .
$
The number of phases is estimated as follows. If the phase terminates because of Condition (i), it has discovered $X$
critical values among those in the current critical interval, which are outside the new critical interval $I^w$.
Hence the number of such phases is at most $L/X$.\footnote{%
  This is a rather weak aspect of the analysis. For the bound $L/X$ to materialize, the $X$ critical values 
  of each phase must form a prefix or a suffix of the sequence of critical values in the current interval 
  $I$, or rather, more precisely, be near the beginning or the end of the sequence. In general, when these
  critical values are more uniformly spread within $I$, the number of such phases should be much smaller. 
  It is a challenging open problem to turn this intuition into an improved procedure, if possible.}

A phase that terminates because of Condition (ii) consumes $Y$ of the total cost of the decision procedure: 
when we pass to the next phase we follow a single path of the current $T$, but all paths use at least $Y$ of the
cost, excluding the work done at the root. Note that, by construction, the part of the simulation performed along 
a path of the tree in one phase, excluding the work performed at the root, is disjoint from the part performed
along a path of the tree in a different phase. This is easily seen to imply that the number of phases of type (ii) is at most $D(n)/Y$.

The sum $\sum_T C_0(T)$ of the initialization costs of all the trees is at most $D(n)$, because these initializations 
perform pairwise disjoint portions of the decision procedure. (Observe that execution at a root is always
completed, no matter how expensive it is.)

We set our parameters so that
$
\frac{L}{X} = \frac{D(n)}{Y} \text{ and } XY = D(n)\log n .
$
That is, we choose 
$
X = L^{1/2}\log^{1/2}n \text{ and } 
Y = D(n)\log^{1/2}n/L^{1/2} ,
$
and obtain that the overall cost of the simulation is $O(L^{1/2}D(n)\log^{1/2} n)$ (this clearly also subsumes the cost $\sum_T C_0(T)$).

In total, the cost of the optimization procedure is
\[
O^*\left(\frac{n^{3/2}}{L^{1/2}} + L^{1/2}D(n) \right) = 
O^*\left(\frac{n^{3/2}}{L^{1/2}} + L^{1/2} n \right).
\]
We balance these two terms by choosing $L = n^{1/2}$, and obtain a total cost of $O^*(n^{5/4})$. 
This result is part of the summary in Theorem~\ref{thm:gen-unweighted2} (which also includes other variants of the problem), given below.

\subsection{Other variants and unit disks}

\smallskip
\noindent
{\bf Scaling by a multiplicative factor.}
Consider first the case of intersection graphs where the radii are scaled by a common multiplicative factor $\lambda>0$. In this case the critical value induced by a pair of disks $D$, $D'$ satisfies
$
\|c_D-c_{D'}\| - \lambda\rho_D - \lambda\rho_{D'} = 0 , 
$ or 
$
{\displaystyle
\lambda = \frac{ \|c_D-c_{D'}\| }{\rho_D + \rho_{D'} } } .
$
As above, the algorithm requires a procedure that computes the number of critical values in an interval
$(\lambda_1,\lambda_2)$. So we convert this task to batched range searching in three dimensions, where the
ranges are now
\[
\sigma_D = \{ D' \mid 
\lambda_1 (\rho_D + \rho_{D'}) \le \|c_D-c_{D'}\| \le \lambda_2 (\rho_D + \rho_{D'}) \} .
\]
These too are (truncated) conical shells, but here the bounding cones have a common apex and different opening angles (see Figure~\ref{fig:shells}(b)).
Other than this new type of ranges, the preceding machinery proceeds verbatim. The interval shrinking runs in 
$O^*(n^{3/2}/L^{1/2})$ expected time and the bifurcation procedure runs in $O^*(L^{1/2}n)$. With a proper choice of $L$, the overall cost is $O^*(n^{5/4})$ expected time.

\smallskip
\noindent
{\bf Proximity graphs.}
Next, consider the case of proximity graphs.
In this case, we assume that the disks of $\D$ are pairwise disjoint, and we are given an additional parameter $r > 0$. (For $r=0$ we get the intersection graph $G^\times$.) The set of edges of the proximity graph $G_r$ consists of all pairs of disks $(D,D')$ for which
$
\dist(D,D') = \|c_D - c_{D'}\| - \rho_D - \rho_{D'} \le r .
$
In the RSP problem for proximity graphs, we seek the smallest value of $r$ for which $G_r$ has the desired $s$-$t$ path. 
It is easy to see that this problem can be reduced to the (additive version of the) corresponding problem for 
intersection graphs, by simply adding $r/2$ to the radius of each of the disks. 
In the optimization procedure, the critical values are of the form $\|c_D - c_{D'}\| - \rho_D - \rho_{D'}$.
Thus, except for the factor $\frac12$ used earlier, the procedure is essentially identical to the earlier one.

\smallskip
\noindent
{\bf Unit disks.}
Finally, consider the special case of unit disks, i.e., where all the radii are equal, and consider the 
intersection graph of the disks. In this case, we get a better bound, since the critical values are merely 
(one half of the) inter-point distances, and hence the interval shrinking step can be performed in 
$O^*(n^{4/3}/L^{1/3})$ randomized expected time, as in \cite{BFKKS}. Modifying the expression for the overall running time accordingly, we get \\
$
{\displaystyle O^*\left(\frac{n^{4/3}}{L^{1/3}} + L^{1/2} n \right) }
$
randomized expected time.
We now balance these two terms by choosing $L = n^{2/5}$, and obtain a total cost of $O^*(n^{6/5})$.
In summary, we have shown:
\begin{theorem} \label{thm:gen-unweighted2}
The reverse shortest path problem for unweighted intersection or proximity graphs of arbitrary disks
in the plane can be solved in $O^*(n^{5/4})$ randomized expected time. This bound applies to all the variants
of the problem listed above. In the case of unit disks, where all radii are equal, the problem can be solved in $O^*(n^{6/5})$ randomized expected time. 
\end{theorem}

\noindent{\bf Remark.}
In the case of unit disks, the decision procedure itself can be implemented to run faster than $O(n\log^4n)$. 
Chan and Skrepetos~\cite{CS} even present an algorithm that runs in linear time (after a preliminary sorting step);
see also~\cite{CabelloJ15}. However, the critical values produced by their procedure are not all inter-point 
distances, or, more precisely put, some of the critical values are determined by triples of input points
rather than pairs. This affects the running time of the optimization procedure. To avoid the use of complex 
dynamic data structures, we modify the algorithm of Chan and Skrepetos, so that it can be combined with the optimization procedure, see Section~\ref{sec:udg_bsf}.

\section{Reverse shortest paths for weighted disk graphs} \label{sec:weighted}

\label{sec:sp-dg}

Recall that, for a set $\D$ of $n$ disks in the plane and an additional parameter
$r\ge 0$, the proximity graph $G_r$ of $\D$ has an edge between every pair of disks such that 
\begin{equation} 
\label{eq:distD}
\dist(D,D') := \|c_D - c_{D'}\| - \rho_D - \rho_{D'} \le r.
\end{equation}
The intersection graph $G^\times$ of $\D$ is the proximity graph of $\D$ for $r=0$, except that in proximity 
graphs we usually assume (but do not have to assume) that the disks are pairwise disjoint.

We next assign weights to the edges.
There are two natural choices for these weights. One is to define the weight of an edge $(D,D')$ as the distance $\|c_D - c_{D'}\|$ between the centers.
The other is to define the weight of $(D,D')$ to be $\dist(D,D')$ if 
$\dist(D,D') \ge 0$ and $0$ otherwise (that is $0$ if the disks intersect). Note that for intersection graphs only the first choice gives a reasonable weight function.

The case of weighted unit disk graphs is a special case of the problem for intersection graphs
$G^\times$ with weights equal to the distances between the centers. In this case we have $\|c_D - c_{D'}\| = \dist(D,D')  + 2$. 

\subsection{The decision procedure}
\label{sec:4.1}

We are given two disks $D_s$, $D_t$ of $\D$, and the proximity graph $G_r$ 
of $\D$ for some $r\ge 0$, weighted by one of the weight functions above.
We  want to compute the (length of the) shortest path $\pi(D_s,D_t)$ in $G$ from $D_s$ to $D_t$. 
We solve this by a clever implementation of Dijkstra's algorithm. The required sophistication of the 
implementation depends on the type of the weight function we use.

\paragraph{The edge weights are the distances between the disks.}
We start with the simpler case in which the weight of
$(D,D')$ is $\dist(D,D')$ (or $0$ if the disks intersect). Specifically, the same function that defines the graph 
via the threshold $r$ (Equation  (\ref{eq:distD})) is also used to define the weights. 

The high-level approach is similar to that proposed by Cabello and Jej\v{c}i\v{c}~\cite{CabelloJ15},
although their original algorithm was given only for intersection graphs of unit disks.
We briefly recall this technique, adapted to our context.

We maintain a decomposition of $\D$ into three disjoint subsets $\R$, $\K$ and $\U$, where $\R\cup\K$ is the set
of disks $D$ for which we already have the correct distance label $\delta(D)$ (the length of the shortest path from $D_s$), and $\U$ is the remainder of $\D$, consisting
of disks whose distance labels have not yet been determined.
$\R$ (resp., $\K$) is the set of \emph{active} (resp., \emph{dead}) disks of $\R\cup\K$, meaning that the 
disks of $\R$ still have outgoing edges in $G_r$ to disks of $\U$, while disks of $\K$ have no such edges. 
Initially, $\R = \{D_s\}$, $\K = \emptyset$, and $\U = \D\setminus\{D_s\}$.

A single step of our implementation of Dijkstra's algorithm, a so-called \emph{Dijkstra step}, picks the closest pair $(D,D')$ in $\R\times\U$, where the modified distance between $D$ and $D'$ is defined as
\begin{equation} \label{eq:w}
\dd(D,D') = \delta(D) + \dist(D,D') .
\end{equation}

Then we need to verify that $(D,D')$ is indeed
an edge of $G_r$. If this is not the case, i.e., $\dist(D,D') > r$, we move $D$ from $\R$ to $\K$, 
and never process $D$ again.
This action is justified by the following variant of~\cite[Lemma 6]{CabelloJ15}:
\begin{lemma} \label{lem6}
Let $D$ and $D'$ be as above, and assume that $\dist(D,D') > r$. Then $\dist(D,D'') > r$ for every disk $D''\in\U$.
\end{lemma}
\begin{proof}
By assumption, and since $(D,D')$ is the closest pair in $\R\times\U$, we have, for each $D''\in\U$,
$
\delta(D) + r < \delta(D) + \dist(D,D') \le
\delta(D) + \dist(D,D'') ,
$
so $\dist(D,D'') > r$, as asserted.
\end{proof}

Note that Lemma \ref{lem6} also applies to weighted unit disk graphs (with weight $\|c_D - c_{D'}\|$ for edge $(D,D')$.)

Assume then that $\dist(D,D') \le r$. We move $D'$ from $\U$ to $\R$, assign to it the distance label 
$
\delta(D') = \delta(D) + \dist(D,D'),
$
and set $\prev(D') = D$, namely $D$ is the disk preceding $D'$ along the shortest path from $D_s$ to $D'$. 
We then repeat the entire step with the new sets $\R$ and $\U$, and continue until $\R$ empties out.
Upon termination, $\U$ is the set of disks that are unreachable from $D_s$ in $G_r$, and $\K$ is the 
set of reachable disks, each with its correct distance label and its predecessor.

The correctness of this procedure is argued exactly as in
\cite{CabelloJ15}, even though the distance function is different. An efficient implementation is obtained by
applying the efficient dynamic bichromatic closest-pair data structure of Kaplan et al.~\cite{KaplanMRSS20}, 
or rather its improved version by Liu~\cite{Liu}, to $\R\times\U$, under the distance
function given in (\ref{eq:w}) and under deletions from $\U$ and insertions into and deletions from $\R$. 
The technique requires the following two properties.

\medskip
\noindent
(i) The Voronoi diagram $\Vor_1(\U)$ of $\U$, under the distance function
$
\dd_1(q,D') = \|q - c_{D'}\| - \rho_{D'} ,
$
for $D'\in\U$, where $c_{D'}$ and $\rho_{D'}$ are the respective center and radius of $D'$, has linear complexity.

\medskip
\noindent
(ii) The Voronoi diagram $\Vor_2(\R)$ of $\R$, under the distance function
$
\dd_2(q,D) = \delta(D) + \|q - c_D\| - \rho_D , 
$
for $D\in\R$, also has linear complexity.

\medskip
Both properties indeed hold, as each diagram is an additively-weighted Voronoi diagram of the set of centers 
of the disks of $\U$ or of $\R$. Both diagrams need to be maintained dynamically, and the techniques of
\cite{KaplanMRSS20,Liu} do that, with a polylogarithmic cost for each update and query operation.
In the more efficient implementation of \cite{Liu}, the amortized cost of an update is $O(\log^4 n)$.
(This holds for deletions; insertions and queries are faster.)

In summary, we have shown:
\begin{theorem} \label{thm:weight1}
The shortest-path tree from some starting disk in the proximity graph of $n$ disks (of
arbitrary radii) in the plane, under the weights of inter-disk distances, can be constructed in $O(n\log^4 n)$ time.
\end{theorem}

Next we address the more challenging weighted case when the edge weights are the distances between the centers.

\paragraph{The edge weights are the distances between the centers.}
In this case we need to maintain a dynamic bichromatic closest pair structure under the distance
$
\lambda(D,D') = \delta(D) + \|c_D - c_{D'}\|,
$
for $D\in\R$ and $D'\in \U$. But then Lemma~\ref{lem6} does not hold anymore, because it is then possible that 
the closest pair $(D,D')\in \R\times\U$ satisfies $\dist(D,D') > r$ but there could be other pairs $(D,D'')$ with $\dist(D,D'') \le r$, so $D$ cannot be removed from $\R$ yet.

To overcome this difficulty we obtain a \emph{bichromatic closest pair} (BCP for short) data structure
by applying a black-box reduction of 
Chan~\cite{Chan:bcp}, which simplifies and slightly improves an earlier technique of Eppstein~\cite{Epp},
to two novel nearest-neighbor data structures that we now describe.

Our new data structures find for $D\in \R$ a nearest neighbor $D'\in \U$ according to the 
distance function $\lambda$ but only among disks $D'$ such that $dist(D,D') \le r$, and similarly for $D\in \U$
and neighbors in $\R$.

We then run the Dijkstra steps, as in the previous case, using the modified BCP data structure, as long as there are still neighbors (in $G_r$) in $\R\times \U$.
Upon termination, $\U$ is the set of disks that are unreachable from $D_s$ in $G_r$, and $\R$ is the 
set of reachable disks, each with its correct distance label and its predecessor.

Our nearest-neighbor data structure consists of two balanced search trees $T_\U$ and
$T_\R$. The tree $T_\U$ ($T_\R$ is handled in a similar manner; see below)
stores the disks $D'\in \U$ at its leaves, sorted in increasing order of the values
$\rho_{D'}$. For each node $v$ of $T_\U$, we maintain an additively-weighted Voronoi diagram $\Vor_1(\U_v)$ on the set $\U_v$ of the
disks of $\U$ stored at the leaves of the subtree rooted at $v$, where the weight of a disk $D'$ is $-\rho_{D'}$.
We also maintain a second standard (unweighted) Voronoi diagram $\Vor_2(\U_v)$ for $\U_v$.

\looseness=-1
Querying $T_\U$ with a disk $D\in\R$ is performed as follows.
We find the leftmost leaf $w_0$ of $T_\U$ whose disk $D'_0$ satisfies 
$
\dist(D,D'_0) = \|c_D - c_{D'_0}\| - \rho_D - \rho_{D'_0} \le r .
$
To find $w_0$, we search in $T_\U$ starting at the root. At each node $u$ that we reach, with a left child $v$ and a right child $w$, we search with $D$ in $\Vor_1(\U_v)$, 
and obtain its nearest neighbor $D'$ in the corresponding subset of disks. 
If $\|c_D - c_{D'}\| - \rho_{D'} \le \rho_D + r$, or, equivalently, $\dist(D,D') \le r$, then we continue 
the search at the left child $v$. Otherwise we continue the search with the right child $w$. 
At the root we first 
search in $\Vor_1({\U_{\rm root}} = \U)$. If the nearest neighbor $D'$ satisfies 
$\|c_D - c_{D'}\| - \rho_{D'} > \rho_D + r$, that is, $\dist(D,D') > r$, 
we conclude that $D$ has no neighbor in $\U$.

Let $w_0$ be the leaf
that the search reaches. Note that if a disk $D'$ is
stored at a leaf to the left of $w_0$ then, by construction, 
$\dist(D,D') > r$. That is, only disks stored to the right of $w_0$, including $w_0$, need to be considered.

The search for $w_0$ provides us with a representation of the set $\U_{w_0}^+$ of the disks to the right of $w_0$ as
the union of $O(\log n)$ pairwise disjoint canonical sets, each being the set of disks stored at the root of some subtree of $T_\U$.
We query with $D$ in each of the $O(\log n)$ diagrams $\Vor_2(\U_v)$ that correspond to these subtrees, and return
the disk that is nearest to $D$, i.e., with the minimum distance between their centers, among all the resulting nearest neighbors.
The correctness of this procedure is a consequence of the following lemma.
\begin{lemma} \label{lem:Utree}
In the above procedure, the output disk $D'$ satisfies $\dist(D,D') \le r$.
\end{lemma}
\begin{proof}
Let $D'_0$ be the disk at the leaf $w_0$. By construction, we have 
\begin{equation} \label{eq:1}
\rho_{D'_0} \le \rho_{D'} .
\end{equation}
By construction, $D'$ is the disk in $\U_{w_0}^+$ nearest to $D$ in the inter-center distance. That is,
\begin{equation} \label{eq:3}
\|c_D - c_{D'}\| \le \|c_D - c_{D'_0}\| .
\end{equation}
Assume by contradiction that $\dist(D,D') > r$. Then, by construction,
\[
\|c_D - c_{D'_0}\| - \rho_{D'_0} - \rho_D =
\dist(D,D'_0) \le r < \dist(D,D') =
\|c_D - c_{D'}\| - \rho_{D'} - \rho_D .
\]
That is,
\begin{equation} \label{eq:2}
\|c_D - c_{D'_0}\| - \rho_{D'_0} < \|c_D - c_{D'}\| - \rho_{D'} .
\end{equation}
Adding (\ref{eq:1}) and (\ref{eq:2}), we get a contradiction to (\ref{eq:3}). This establishes the lemma.
\end{proof}

The tree $T_\R$ is defined in an analogous manner for the disks of $\R$, except that
(a) the leaves are sorted in increasing order of $\delta(D) + \rho_{D}$, and
(b) the Voronoi diagram $\Vor_2(\R_v)$ at a node $v$ of $T_\R$, where $\R_v$ is the set of disks stored 
at the root of the subtree rooted at $v$, is the additively-weighted diagram on $\R_v$,
where the additive weight of a disk $D$ is $\delta(D)$.

The first kind of Voronoi diagrams $\Vor_1(\R_v)$ are defined exactly as in the case of $T_\U$, with the additive weight being $-\rho_D$.

Here too, when we query $T_\R$ with a disk $D'$ of $\U$,
we search in the tree to find the leftmost leaf $w_0$ of $T_\R$ whose disk $D_0$ satisfies 
$
\dist(D_0,D') = \|c_{D_0} - c_{D'}\| - \rho_{D_0} - \rho_{D'} \le r .
$
This is performed as follows. We first search for the leftmost leaf $w_0$ whose associated disk $D_0$ 
satisfies $\dist(D_0,D') \le r$, using the same technique as in the previous search in $T_\U$, 
in which we query various Voronoi diagrams $\Vor_1(\R_v)$ along the search path to $w_0$.
We then obtain the set $\R_{w_0}^+$ of all disks stored at the leaves to the right of $w_0$, including $w_0$, as the disjoint union of $O(\log n)$ subtrees. 
We query each of the diagrams $\Vor_2(\R_v)$ associated with these subtrees, and return the disk $D$ that is the nearest neighbor to $D'$ over all these diagrams.

The correctness of this procedure is a consequence of the following `sister' lemma to Lemma~\ref{lem:Utree}.
\begin{lemma} \label{lem:Rtree}
In the above procedure, the output disk $D$ satisfies $\dist(D,D') \le r$.
\end{lemma}
\begin{proof}
Let $D_0$ be the disk at the leaf $w_0$. By construction, we have 
\begin{equation} \label{eq:R1}
\delta(D_0) + \rho_{D_0} \le \delta(D) + \rho_{D} .
\end{equation}
By construction, $D$ is the disk in $\U_{w_0}^+$ nearest to $D'$ in the inter-center distance with the additive weight $\delta$. That is,
\begin{equation} \label{eq:R3}
\delta(D) + \|c_D - c_{D'}\| \le \delta(D_0) + \|c_{D_0} - c_{D'}\| .
\end{equation}
Assume by contradiction that $\dist(D,D') > r$. Then, by construction,
\[
\|c_{D_0} - c_{D'}\| - \rho_{D_0} - \rho_{D'} =
\dist(D_0,D') \le r < \dist(D,D') =
\|c_D - c_{D'}\| - \rho_{D} - \rho_{D'} .
\]
That is,
\begin{equation} \label{eq:R2}
\|c_{D_0} - c_{D'}\| - \rho_{D_0} < \|c_D - c_{D'}\| - \rho_{D} .
\end{equation}
Adding (\ref{eq:R1}) and (\ref{eq:R2}), we get 
\[
\delta(D_0) + \|c_{D_0} - c_{D'}\| < \delta(D) + \|c_D - c_{D'}\| ,
\]
which is a contradiction to (\ref{eq:R3}). This establishes the lemma.
\end{proof}

We make these nearest-neighbor data structures dynamic by maintaining the various Voronoi diagrams 
$\Vor_1(\U_v)$, $\Vor_2(\U_v)$, $\Vor_1(\R_v)$, $\Vor_2(\R_v)$, dynamically, using the technique of 
\cite{KaplanMRSS20,Liu}, in which the maximum update time of each diagram is $O(\log^4 n)$.
An update of either of the trees $T_\U$, $T_\R$ requires updating $O(\log n)$ data structures along 
the path to the leaf containing the inserted or deleted element, and therefore takes $O(\log^5 n)$ time.
The transformation of Chan\cite{Chan:bcp} from the two
nearest-neighbor structures to a BCP data structure (see also the earlier work of Eppstein~\cite{Epp}) incurs an additional logarithmic overhead. Overall we get

\begin{theorem} \label{thm:weight2}
The shortest-path tree from some starting disk in the proximity graph of $n$ disks (of arbitrary radii) 
in the plane, under the weights of inter-center distances, can be constructed in $O(n\log^6 n)$ time.
\end{theorem}

\medskip
\noindent{\bf Remark.}
As in the unweighted version, in the case of unit disks, the decision procedure itself can be implemented 
to run faster than $O(n\log^6n)$. For example, Wang and Xue~\cite{WangX20} present an algorithm that runs 
in $O(n \log^2 n)$ time, but this algorithm is not suitable for our optimization procedure, since its comparisons 
generate critical values that are determined by more than two disks. We thus replace its ``problematic'' 
components by new ones, to obtain a (somewhat simpler) algorithm that is suitable for the optimization procedure; see Section~\ref{sec:udg_dijk}.

\subsection{The optimization procedure}
\label{sec:w_opt}

The most natural reverse shortest path question is to find the smallest value of $r$ such that the length 
of the shortest path in $G_r$ from $D_s$ to $D_t$ is smaller than some given real parameter $w$.

This optimization procedure is implemented as in the unweighted case, using the same combination of the 
interval shrinking and bifurcation procedures. The critical values are the same as in the case for the unweighted proximity graph. 

In summary, we have shown:
\begin{theorem} \label{thm:gen-uweighted2}
The reverse shortest path problem for weighted intersection or proximity graphs of arbitrary disks
in the plane can be solved in $O^*(n^{5/4})$ randomized expected time. This bound applies to all the variants
of the problem listed above. In the case of unit disks, where all radii are equal, the problem can be solved in $O^*(n^{6/5})$ randomized expected time. 
\end{theorem}

\medskip
\noindent{\bf Remark.}
The machinery developed in this section seems to be more broadly applicable to other optimization questions 
that involve shortest paths in weighted proximity or intersection graphs, of disks or of more general 
geometric objects. Simple extensions could involve different definitions of proximity or 
the use of other weight functions.\footnote{%
  For the decision procedure to retain near-linear time, we need to ensure that the resulting (generalized) Voronoi diagram continues to have linear complexity, which does hold in many, but not all, applications.}
Other extensions could involve different problems, such as computing 
all-pairs shortest paths, computing shortest paths with negative edge weights, computing the diameter of the graph, etc.

\section{BFS and Dijkstra's algorithm in unit disk graphs}
\label{sec:udg_bfs_dijk}

One notable drawback of the decision procedures presented in Sections~\ref{sec:unweighted} and~\ref{sec:weighted}, especially from a practical point of view, is that they use a data structure for maintaining additively-weighted Voronoi diagrams under insertions and deletions, which is a rather complex task. Moreover, this also requires some additional care in the implementation of the bifurcation procedure, as discusses earlier. It is therefore desirable to avoid using this data structure, and dynamic data structures in general.

The goal of this section is to present such decision procedures for the special and important case of unit disks. That is, we wish to develop efficient and practical algorithms for BFS and Dijkstra's algorithm in unit disk graphs. This is a well-studied topic and such algorithms already exist, as surveyed in the introduction. However, they are not suitable for our purpose, since the efficiency of the bifurcation procedure, which simulates the execution of the decision
procedure at $r^*$, depends on the assumption that whenever it bifurcates on a concrete
value $r$, this value is a critical value that is determined by just a pair of points, like the inter-point distance between them. We thus
develop new decision algorithms, which essentially replace
some components of the known algorithms with new (non-trivial) ones, so that the
required assumption holds. We believe that the new algorithms, especially the one for the weighted version, are also somewhat simpler than the original ones on which they are based.

We stress that in this section we only consider the decision problems. The corresponding optimization procedures
are implemented more or less as before, with some obvious modifications, except for one technical issue that concerns
the grid construction, a key ingredient of the technique presented here. This issue is elaborated at the end of the section.

In this section $P$ is a set of $n$ points in the plane, $s$ and $t$ are two designated points in $P$, and $G_r$ is the graph $(P,E_r)$, where the edges in $E_r$ are all the pairs $(p,q)$ such that $\|p-q\| \le r$ (it is the intersection graph of congruent disks of radius $r/2$ centered at the points of $P$). In the weighted version, the weight of an edge $(p,q) \in E_r$ is $\|p-q\|$.

\subsection{Unweighted unit disks}
\label{sec:udg_bsf}

The simplest way of implementing the decision procedure is to run a BFS on $G_r$
from $s$, and check that we reach $t$ after at most $k$ levels. However, since $G_r$
may contain quadratically many edges, some more careful implementation of the BFS is needed.
Our algorithm is similar in spirit to that of Chan and Skrepetos~\cite{CS}, but differs in a few key details that 
make the optimization algorithm more efficient. It is also similar to the algorithm of 
Cabello and Jej\v{c}i\v{c}~\cite{CabelloJ15}, except that they use the Delaunay triangulation of $P$ instead of a grid. 

We construct a grid of cell size $r/\sqrt{2}$ and distribute the points of $P$ among its cells.\footnote{%
  A slightly different implementation is used in the optimization procedure; see a comment at the end of the section.}
This takes $O(n)$ time. For each grid cell $\tau$ denote by $P_\tau$ the set of points of $P$ 
in $\tau$. We perform a BFS from $s$ in $G_r$ as follows. Let $P^i$ be the set of points at 
level $i$ of the BFS. We set $P^0=\{s\}$. Given $P^i$ we compute $P^{i+1}$ as follows. 
For each grid cell $\tau$ such that $P^i_\tau := P^i\cap P_\tau \not=\emptyset$ we compute 
the (standard, static) Voronoi diagram $\Vor(P^i_\tau)$. Define the \emph{neighborhood} $\NN(\tau)$ of $\tau$ 
to consist of all the grid cells whose row and column indices are at most $2$ apart from the 
respective row and column indices of $\tau$, excluding $\tau$. (All the cells that contain 
points at distance $\le r$ from $\tau$, other than those in $\tau$ itself, belong to $\NN(\tau)$.) 
Then for each grid cell $\tau'$ in $\NN(\tau)$ we locate each $p\in P_{\tau'}$ in 
$\Vor(P^i_\tau)$, thereby obtaining the nearest neighbor of $p$ in $P^i_\tau$. (Actually, 
it suffices to locate only those points of $P_{\tau'}$ that have not yet been reached 
by the BFS, but this does not affect the asymptotic complexity of the procedure.)
If the distance between $p$ and its nearest neighbor is $\le r$ and $p$ has not yet been 
reached by the BFS, we add $p$ to $P^{i+1}$. Note that there is no need to
carry out this step for $\tau' = \tau$ (which is excluded from $\NN(\tau)$ anyway), 
since all the points of $P_\tau$ are at distance at most $r$ from any point of $P^i_\tau$, so all the points of 
$P_\tau\setminus P^i_\tau$ are immediately added to $P^{i+1}$.

Once we obtain $P^k$ we stop the BFS and answer positively if $t$ has been reached by the BFS, and negatively if it has not. (We can stop earlier once $t$ has been reached.)

This decision procedure takes $D(n) := O(n\log n)$ time.\footnote{%
  Chan and Skrepetos's procedure runs in linear time, after pre-sorting the points by their
  $x$- and $y$-coordinates, but (a) Chan and Skrepetos's procedure is not suitable for the optimization procedure,
  for the reasons discussed above, and (b) such an improvement has negligible effect on the performance
  of our optimization algorithm anyway.}
To see this, put $n_\tau := |P_\tau|$ and $n^i_\tau := |P^i_\tau|$, for each $\tau$ and $i$. 
The cost of processing the $i$-th step of the BFS (moving from $P^i$ to $P^{i+1}$) is
\begin{equation} \label{eq:icost}
O\Bigl( \sum_{\tau \,\mid\, n^i_\tau \ne 0} 
\Bigl( \Bigl( n^i_\tau + \sum_{\tau'\in \NN(\tau)} n_{\tau'} \Bigr) \log n \Bigr) \Bigr),
\end{equation}
and the overall cost is the sum of these bounds over $i$. The asserted bound on the cost
is a consequence of the following two observations:
(i) Each cell $\tau$ can have $n^i_\tau \ne 0$ for at most two consecutive values of $i$,
the $i$ at which the first point of $P_\tau$ is reached by the BFS 
(and added to $P^{i+1}$) and the next index $i+1$.
(ii) Each cell $\tau'$ is a neighbor of only $O(1)$ cells $\tau$, and can therefore
participate in the inner sum in (\ref{eq:icost}) for at most two indices $i$ for each neighbor cell
$\tau$, that is, for a total of at most $O(1)$ indices.

We have thus shown:
\begin{theorem}
	Given $P$, $s$, $t$, $k$ and $r$ as above, we can determine whether the shortest path from 
	$s$ to $t$ in $G_r$ is of length at most $k$, in $O(n\log n)$ time.
\end{theorem}

\subsection{Weighted unit disks}
\label{sec:udg_dijk}

Let $P$, $s$, $t$, $r$ and $w$ be as above. The goal is to determine whether the length of the 
weighted shortest path in $G_r$ from $s$ to $t$ is at most $w$. 

As in the unweighted case, we construct a grid of cell size $r/\sqrt{2}$, and distribute the
points of $P$ among its cells,\footnote{%
	This step is slightly modified in the optimization algorithm, as in the unweighted case,
	to ensure that the grid size does not explicitly depend on the unknown $r^*$. Again, see a comment at the end of the section.}
denoting by $P_\tau$ the set of points of $P$ in cell $\tau$. Similar to what has been proposed 
in Wang and Xue~\cite{WangX20}, we run a modified version of Dijkstra's algorithm on $G_r$, starting 
from $s$, with several key changes that aim to overcome the difficulty that $G_r$
may have too many edges (up to quadratic in the worst case).

The high-level approach is the standard one: We maintain the set $\R$ of vertices that the algorithm
has reached (initialized to $\{s\}$), whose shortest paths from $s$ have already been computed, 
and the set $\U$ of vertices that have not yet been reached (the use of the set $\K$ of the dead points is not 
needed in this version). With each $a\in P$ we store its distance label $\dist(a)$. 
For $a\in \R$, this is the length of the shortest path from 
$s$ to $a$, and for $a\in \U$ it is the length of the shortest path to $a$ discovered so far.
We also store a pointer $\prev(a)$ to the point preceding $a$ on the path to $a$ (as above,
the true shortest path for $a\in \R$, and the shortest path discovered so far for $a\in \U$).

As long as $\U$ is not empty, we pick its element $v$ with minimal distance label $\dist(v)$, 
update the current distance labels of the (unreached) neighbors of $v$ (by the so-called 
{\sc Relax} operation), add $v$ to $\R$ and remove it from $\U$. However, we do more, by 
applying this step to all the points in the grid cell of $v$, and by moving all of them 
from $\U$ to $\R$. Moreover, the implementation of the steps of finding the unreached vertex 
with minimal distance label and of updating points in its neighborhood differs from the standard 
one, since we cannot afford to process all the edges of $G_r$. The high-level approach follows
that of Wang and Xue~\cite{WangX20}, but differs from it in certain aspects which simplify it considerably 
and make it suitable for the optimization algorithm, presented in Section~\ref{sec:w_opt}.

\paragraph*{Notation.}
For $a\in P$ define $\tau_a$ to be the grid cell containing $a$. Define $\NN(\tau)$ to be the
set of the $O(1)$ neighboring cells of $\tau$, which form a $5\times 5$ portion of the grid
centered at $\tau$, as defined for the unweighted case, but this time $\NN(\tau)$ contains 
$\tau$ too. For a subset $A\subseteq P$ and a grid cell $\tau$, put $A_\tau := A\cap\tau$.

\paragraph*{The modified Dijkstra's algorithm.}

Repeating and expanding the terminology in the preceding overview, the algorithm proceeds as follows.
It maintains a map $\dist$, the distance label, that maps each point $p\in P$ to the 
length of the shortest path from $s$ to $p$ that has been discovered so far. It ensures that, 
when $p$ is added to $\R$, $\dist(p)$ is equal to the length of the shortest path from $s$ to 
$p$, which we denote as $\optdist(p)$. The algorithm updates $\dist$ using a procedure $\UPDATE(B,C)$, 
where $B$ and $C$ are subsets of $P$ (actually of $\U$). Informally, we say that in this invocation, 
$B$ is updated according to $C$. The procedure performs, for each $u\in B$, the update
\[
\dist(u) := \min \Bigl\{ \dist(u),\; \min \left\{ \dist(v) + \|u-v\| \mid v\in C,\; \|u-v\| \le r \right\} \Bigr\} .
\]
Specific details about an efficient implementation of the procedure are given below. 

The algorithm of Wang and Xue~\cite{WangX20} processes the grid cell by cell. It processes
the cell containing the point of minimum distance label, which has just been moved from $\U$ to $\R$,
and then deletes from $\U$ all the points in that cell. Concretely, the algorithm proceeds as follows.

\medskip
\noindent
{\tt Initialize:} $\U := P;\quad \dist(s) := 0;\quad ({\tt forall}\; a \in P\setminus\{s\}) \; \dist(a) := \infty$. \\
({\tt while} $\U \ne\emptyset$) \\
\hspace*{1cm}(1) Find $v\in \U$ with the minimum value of $\dist(v)$. \\
\hspace*{1cm}(2) For each $a\in \U_{\tau_v}$, update $\dist(a)$ according to $\U_{\NN(\tau_v)}$,
by calling $\UPDATE\left(\U_{\tau_v}, \U_{\NN(\tau_v)}\right)$. \\
\hspace*{1cm}(3) For each $b\in \U_{\NN(\tau_v)}$, update $\dist(b)$ according to $\U_{\tau_v}$, 
by calling $\UPDATE\left(\U_{\NN(\tau_v)}, \U_{\tau_v}\right)$. \\
\hspace*{1cm}(4) Remove $\U_{\tau_v}$ from $\U$ (and add it to $\R$). \\
{\tt end while} 

\medskip
\noindent
Here $\U_{\NN(\tau_v)}$ is a shorthand notation for $\bigcup_{\tau\in\NN(\tau_v)} \U_\tau$. 

The correctness of the algorithm is a consequence of the following three claims. We present (our 
versions of) the proofs of these claims, originally given in \cite{WangX20}, for the sake of completeness.

\begin{claim} \label{claim1}
	Every point $x\in \R$ has a correct distance label, namely, $\dist(x) = \optdist(x)$.
\end{claim}

\begin{claim} \label{claim2}
	The point $v$ of $\U$ with the minimum distance label at the beginning of a Dijkstra step has a 
	correct distance label. Closer points to $s$ have already been deleted from $\U$.
\end{claim}

\begin{claim} \label{claim3}
	All the immediate predecessors (on the true shortest paths) of points in the grid cell 
	of $v$ have correct distance labels at the beginning of the Dijkstra step of $v$.
\end{claim}
\medskip
\noindent{\bf Proof of the claims.}
We establish the invariants in Claims~\ref{claim1}--\ref{claim3} by induction on the 
Dijkstra steps. The invariants hold trivially (some vacuously) at the start of execution. 
Assuming that they hold at the beginning of the $i$-th iteration, for some $i$, we argue 
that they hold at the end of that iteration, that is at the beginning of iteration $i+1$. 
Let $v$ be the point chosen at iteration $i$, namely 
the point of $\U$ with the smallest distance label at this step, and put $\tau = \tau_v$.

Invariant \ref{claim1} (at the end of the iteration) follows from Claims~\ref{claim2} 
and \ref{claim3} and the update rules. 
Indeed, the points deleted from $\U$ at iteration $i$ are those that lie in $\tau$. 
For each such point $v'$, its predecessor $p(v')$ necessarily lies in $\NN(\tau)$. By the induction hypothesis
for Claim~\ref{claim3}, $p(v')$ has a correct distance label at the start of iteration $i$, 
and the first update step ensures that $v'$ gets a correct distance label, as claimed.

For invariant \ref{claim2}, let $x$ be the point with minimum distance label in $\U$
(at the end of iteration $i$). By step (4), $x\notin \tau$.
Assume, by contradiction, that $\optdist(x) < \dist(x)$. 
Consider the (true) shortest path $\pi(x)$ from $s$ to $x$. 
Let $z$ be the last point removed from $\U$ on $\pi(x)$ when we walk along the path from 
$s$ to $x$, i.e., $z$ has been removed at some iteration $j\le i$ 
(it may have been moved at step (4) of iteration $i$). By the induction hypothesis for Claim~\ref{claim1}, 
$z$ has a correct distance label (by the preceding argument, this also holds when $z$ 
has just been removed), and the second update step (3) at iteration $j$ ensures 
that the successor $z^+$ of $z$ on $\pi(x)$ has a correct distance label at the end
of iteration $j$. By our assumption, $z^+$ cannot be $x$ itself, so its distance
label is smaller than $\dist(x)$. Since $z^+$ has not yet been deleted from $\U$, 
we get a contradiction to the minimality of $\dist(x)$, and the first part of 
Claim~\ref{claim2} follows.

For the second part, let $x'$ be a point closer to $s$ than $x$. Assume to the contrary 
that $x'$ has not yet been deleted from $\U$.  By assumption and the definition of $x$, we thus have
\[
\optdist(x') < \optdist(x) \le \dist(x) \le \dist(x') ,
\]
that is, the true distance of $x'$ from $s$ is strictly smaller than its distance 
label. Consider now the shortest path $\pi(x')$ to $x'$, and apply the same 
reasoning as above to obtain a contradiction. This completes the induction step for Claim~\ref{claim2}. 

For Invariant \ref{claim3}, let $x$ be the point in the preceding argument ($x$ will be
the point $v$ of minimum distance label in $\U$ at the beginning of the next iteration).
Let $y$ be a point in the grid cell $\tau_x$ of $x$, and let $y' = \prev(y)$ be the predecessor 
of $y$ on the shortest path $\pi(y)$ from $s$ to $y$. We claim that the predecessor 
$y'' = \prev(y')$ of $y'$ on $\pi(y)$ has already been deleted from $\U$. This would imply, 
by construction, that $y'$ has a correct distance label, which is what we need to show. 
By the triangle inequality, $\|y''-y\| < \|y''-y'\|+\|y'-y\|$ (in general position, we 
do not have an equality). If the right-hand side were $\le r$, $y$ would have been a 
neighbor of $y''$ in $G_r$, and then $y''$ should have been the predecessor of $y$ on 
$\pi(y)$, rather than $y'$. We thus have

\begin{align*}
\optdist(y'') + r & < \optdist(y'') + \|y''-y'\|+\|y'-y\| \\
& = \optdist(y) \le \optdist(x) + \|x-y\| \le \optdist(x) + r.
\end{align*}
It follows that $\optdist(y'') < \optdist(x)$, so, by Invariant \ref{claim2}, $y''$ has 
already been removed from $\U$, implying that $y'$ already has the correct distance label, as claimed.

We have thus established the induction step for all three claims, which in turn completes the proof of the claims.
$\Box$

\paragraph*{Implementing $\UPDATE$.}

To perform the $\UPDATE$ operation at Step (2) of the algorithm, we construct the additively-weighted 
Voronoi diagram $\Vor\left(\U_{\NN(\tau_v)}\right)$, where the weight of a point is its distance label. 
We then search in the diagram with each $a\in \U_{\tau_v}$, and obtain its nearest weighted neighbor $b$.
We perform a {\sc Relax}, by assigning $\dist(a) :=  \min\Bigl\{ \dist(a),\; \dist(b) + \|a-b\| \Bigr\}$, 
and also set $\prev(a) := b$, if the distance label has changed.

To argue the correctness of this procedure, we need to show that, for each $a$ and $b$ as above,
$\|a-b\| \le r$, as this will imply that the Voronoi nearest neighbor is also the nearest neighbor 
(in the path-length sense) in $G_r$. That is, $\dist(b) + \|a-b\|$ is the minimum value of 
$\dist(b') + \|a-b'\|$, over all $G_r$-neighbors $b'$ of $a$ (clearly, any such neighbor must
belong to $\U_{\NN(\tau_v)}$). For this we use the following lemma.

\begin{lemma} \label{lem:near}
	Let $q$ be a query point, and let $A$ be a subset of $P$ with the property that the element 
	$u\in A$ with the smallest distance label satisfies $\|q-u\| \le r$. Then the nearest weighted
	Voronoi neighbor $w$ of $q$ in $A$ also satisfies $\|q-w\|\le r$.
\end{lemma}
\medskip
\noindent{\bf Proof.}
Suppose to the contrary that $\|q-w\| > r$, so $\|q-w\| > \|q-u\|$. 
By assumption, we also have $\dist(u) \le \dist(w)$. Hence
\[
\dist(u) + \|q-u\| < \dist(w) + \|q-w\| ,
\]
contradicting the assumption that $w$ is the nearest weighted Voronoi neighbor of $q$ in $A$.
$\Box$

The asserted claim then follows, with $q, u, w$ set to $a, v, b$, respectively,
since every pair of points in $\tau_v$ lie at distance at most $r$ apart.\footnote{%
  The argument also holds when $v$ is the only point in $\tau_v$.}

Note the subtle issue that, since $\U_{\NN(\tau_v)}$ contains $\U_{\tau_v}$, the nearest weighted 
neighbor of a point $a\in \U_{\tau_v}$ might be $a$ itself. In this case the resulting minimum 
weighted distance to $a$ would be $\dist(a) + \|a-a\| = \dist(a)$. This simply means that no
neighbor of $a$ in $\U_{\NN(\tau_v)}$ can improve its distance label.

The second $\UPDATE$ call, at Step (3) of the algorithm, is harder to analyze. The main difficulty 
is that now the nearest weighted Voronoi neighbor in $U_{\tau_v}$ of a point $a\in \U_{\NN(\tau_v)}$ 
might lie at distance $> r$ from $a$. To address this issue, we apply Lemma~\ref{lem:near} as follows. 
(Note the similarity of this approach to the technique used in Section~\ref{sec:4.1} for the case where the 
edge weights are the distances between the disk centers.)
We sort the points of $\U_{\tau_v}$ in increasing order of their distance labels, and store them in 
this order at the leaves of a balanced binary tree $Q$. For each node $\xi$ of $Q$, we construct the 
additively weighted diagram $\Vor(\U^{(\xi)})$, where $\U^{(\xi)}$ is the set of points of $\U_{\tau_v}$
stored at the leaves of the subtree rooted at $\xi$, and the standard unweighted Voronoi diagram 
$\Vor_0(\U^{(\xi)})$. We store both diagrams at $\xi$. 

Given a query point $q\in U_{\NN(\tau_v)}$, we first search for the leftmost leaf of $Q$, representing
a point $w\in \U_{\tau_v}$, that satisfies $\|q-w\| \le r$. To do so, we first find the point nearest 
to $q$ in $\Vor_0(\U^{(\rm root)})$. If its distance from $q$ is greater than $r$, we stop the search 
and report that $q$ has no neighbors in $G_r$ in $\U_{\tau_v}$. Otherwise, we search in $Q$ as follows. When 
we are at a node $\xi$, with a left child $\eta$ and a right child $\zeta$, we locate $q$ in 
$\Vor_0(\U^{(\eta)})$. If the nearest neighbor to $q$ in that diagram is at distance at most $r$, we 
continue the search at the left child $\eta$, and otherwise we continue the search at the right child $\zeta$.

Let $\xi$ be the leaf at which the search terminates. Clearly, it is the leftmost leaf of $Q$ satisfying
the above property. We represent the set $B$ of leaves of $Q$ that lie to the right of $\xi$, including 
$\xi$ itself, as the disjoint union of $O(\log n)$ canonical subsets, each being the set of leaves of 
some suitable subtree of $Q$. We now locate $q$ at each of the weighted Voronoi diagrams stored at the 
roots of these subtrees, and obtain a corresponding nearest weighted Voronoi neighbor of $q$ in that 
diagram. We retain only those neighbors that lie at distance at most $r$ from $q$, and return the point 
$v$ among the surviving neighbors with the smallest weighted Voronoi distance to $q$. It easily follows
that Lemma~\ref{lem:near} guarantees that $\|q-v\| \le r$, since the element of $B$ with the 
smallest distance label, namely $w$, satisfies $\|q-w\|\le r$ by construction.

The overall cost of the search is $O(\log^2n)$.

An iteration of the while loop in the algorithm thus
takes 
\[
O\left( \left( |\U_{\tau_v}| + |\U_{\NN(\tau_v)}| \right)\log^2 n\right)
\]
time, where $v$ is the point of $\U$ with the smallest distance label.
Each cell $\tau$ can play the role of $\tau_v$ in at most one iteration of the loop, 
because after that iteration all its points are deleted from $\U$. A cell $\tau$ can
participate in a neighborhood $\NN(\tau_v)$ in only $O(1)$ iterations, once for each cell $\tau'$ in
its neighborhood. It follows that the overall cost of the procedure is $O(n\log^2n)$.

The above analysis implies the correctness and efficiency of the whole procedure. That is, we have shown:
\begin{theorem} \label{thm:dec}
	Given $P$, $s$, $t$, $w$ and $r$ as above, we can determine whether the shortest path from 
	$s$ to $t$ in $G_r$ is of length at most $w$, in $O(n\log^2n)$ time.
\end{theorem}

\paragraph{The optimization procedure.}
The optimization procedure is implemented in much the same way as before. 
One notable new ingredient is the grid construction. The distribution of the points amid the grid cells is 
implemented using the floor function, which is not comparison based and thus does not fit the interval 
shrinking methodology. Instead we use the following approach.

We note that, in the unweighted case, $r^*$ must satisfy $\frac1k \|t-s\| \le r^* \le \|t-s\|$, where the left
inequality follows from the triangle inequality. Similarly, in the weighted case $r^*$ must satisfy 
$\frac1n \|t-s\| \le r^* \le \|t-s\|$, where again the left
inequality follows from the triangle inequality.
Say we are in the weighted case (the analysis in the unweighted case is similar).
We fix some sufficiently small $\eps>0$, and construct a sequence of values 
$r_1 = \frac1k \|t-s\|, r_2,\ldots, r_\mu = \|t-s\|$, so that $r_{j+1} \le (1+\eps)r_j$, 
for $j=1,\ldots,\mu-1$, and $\mu = O\left( \frac{1}{\eps}\log k \right)$. We run a binary 
search through the sequence of the $r_j$'s, using the (unsimulated) decision procedure to 
guide the search, and thereby determine, in $O(n\log n\log\mu) = O(n\log n\log\log n)$ time, 
the index $j$ satisfying $r_j \le r^* < r_{j+1} \le (1+\eps)r_j$.

We construct a grid of cell size $r_j/\sqrt{2}$ and distribute the points of $P$ among its cells. 
For a sufficiently small choice of $\eps$, the grid has similar properties as the grid of cell 
size $r^*/\sqrt{2}$, namely, all the cells that contain points of $P$ at distance at most $r^*$ 
from a cell $\tau$ lie in $N(\tau)$. Note that this part of the modified decision procedure does 
not depend explicitly on $r^*$, and so requires no simulation, and it runs in $O(n)$ time. 

Except for this technical issue, the procedure runs in the same manner as before, using interval shrinking 
and bifurcation, and its running time is $O^*(n^{6/5})$, as before.

\bibliography{rsp_udg}

\end{document}